\documentclass[10pt]{article}
\usepackage{fullpage}
\usepackage[bookmarks]{hyperref}
\usepackage{amssymb}
\usepackage{amsmath}
\usepackage{amsthm}
\usepackage[abs]{overpic}
\usepackage{tcolorbox}
\usepackage{graphicx,color,colordvi}
\usepackage{bbm}
\usepackage{cleveref}
\usepackage{stmaryrd}
\usepackage[utf8]{inputenc}
\usepackage{fullpage}
\usepackage[blocks]{authblk}
\usepackage{dsfont}
\usepackage{mathtools}
\usepackage{authblk}

\usepackage{centernot}
\usepackage{pgfplots}

\DeclareMathOperator*{\argmin}{arg\,min}

\def\openone{\leavevmode\hbox{\small1\kern-3.8pt\normalsize1}}

\DeclarePairedDelimiter\ceil{\lceil}{\rceil}

\def\11{\mathbb{I}}

\newtheorem{definition}{Definition}[section]
\newtheorem{proposition}[definition]{Proposition}
\newtheorem{lemma}[definition]{Lemma}

\newtheorem{theorem}[definition]{Theorem}
\newtheorem{corollary}[definition]{Corollary}

\newcommand{\SC}{n^{\star}}
\newcommand{\JS}{\operatorname{JS}}
\newcommand{\LDP}{\operatorname{LDP}}
\newcommand{\sym}{\operatorname{B}}
\newcommand{\asym}{\operatorname{PF}}

\newcommand{\tr}{\mathop{\rm Tr}\nolimits}

\usepackage{pst-node}
\usepackage{tikz-cd}

\newcommand{\cA}{{\cal A}}
\newcommand{\cB}{{\cal B}}

\newcommand{\cN}{{\cal N}}

\newcommand{\cR}{{\cal R}}

\newcommand{\cM}{{\mathcal{M}}}

\def\d{\mathrm{d}}

\usepackage{graphicx}
\usepackage{setspace}
\usepackage{verbatim}
\usepackage{subfig}

\numberwithin{equation}{section}

\DeclareRobustCommand\openone{\leavevmode\hbox{\small1\normalsize\kern-.33em1}}

\newcommand{\id}{{\rm{id}}}
\newcommand{\be}{\begin{equation}}
	\newcommand{\ee}{\end{equation}}
\newcommand{\bea}{\begin{eqnarray}}
	\newcommand{\eea}{\end{eqnarray}}
\newcommand{\beas}{\begin{eqnarray*}}
	\newcommand{\eeas}{\end{eqnarray*}}

\DeclareFontFamily{U}{mathx}{\hyphenchar\font45}
\DeclareFontShape{U}{mathx}{m}{n}{<-> mathx10}{}
\DeclareSymbolFont{mathx}{U}{mathx}{m}{n}
\DeclareMathAccent{\widebar}{0}{mathx}{"73}
\newcommand{\Renyi}{R{\'e}nyi~}

\renewcommand{\d}{\textnormal{d}}

\newcommand{\Id}{{\mathds{1}}}
\DeclareMathAccent{\widehat}{0}{mathx}{"70}
\DeclareMathAccent{\widecheck}{0}{mathx}{"71}

\title{Sample Complexity of Locally Differentially Private \\ Quantum Hypothesis Testing}
\author{Hao-Chung Cheng$^{3,4,5}$, Christoph Hirche$^{1}$, Cambyse Rouz\'e$^{2}$}
\affil{$^1$\textit{Institute for Information Processing (tnt/L3S), Leibniz Universit\"at Hannover, Germany}\\
$^2$\textit{Inria, T\'{e}l\'{e}com Paris - LTCI, Institut Polytechnique de Paris, 91120 Palaiseau, France}
\\ 
\textit{Zentrum Mathematik, Technische Universit\"{a}t M\"{u}nchen, 85748 Garching, Germany}\\
$^3$\textit{Department of Electrical Engineering,
National Taiwan University
, Taipei 10617, Taiwan (R.O.C.)
}\\
$^4$\textit{Physics Division, National Center for Theoretical Sciences, Taipei 10617, Taiwan (R.O.C.)}\\	
$^5$\textit{Hon Hai (Foxconn) Quantum Computing Centre, New Taipei City 236, Taiwan (R.O.C.)}}

\begin{document}

\maketitle

\begin{abstract}
Quantum state discrimination is an important problem in many information processing tasks. In this work we are concerned with finding its best possible sample complexity when the states are preprocessed by a quantum channel that is required to be locally differentially private. To that end we provide achievability and converse bounds for different settings. This includes symmetric state discrimination in various regimes and the asymmetric case. On the way, we also prove new sample complexity bounds for the general unconstrained setting. An important tool in this endeavor are new entropy inequalities that we believe to be of independent interest. 
\end{abstract}

\section{Introduction}

Hypothesis testing is a fundamental primitive in information theory. The most basic setting is that of state discrimination where we are given a quantum state that is either in the state $\rho$ or in the state $\sigma$. The goal is to identify which of these is the case. The corresponding error probability, assuming symmetric priority of the types of errors, is given by
\begin{align}\label{Eq:Error-prob}
        p_e(\rho,\sigma,p) & = p - p E_\frac{1-p}{p}(\rho^{\otimes n}\|\sigma^{\otimes n}),
\end{align}
where $E_\gamma(\rho\|\sigma):=\tr(\rho-\gamma\sigma)_+$ is the Hockey Stick divergence between the states $\rho$ and $\sigma$. The chance of identifying the state correctly can be improved by considering the availability of $n$ copies of the state which allows to measure them jointly. At this point several figures of merit can become useful. If one simply aims to minimize the probability of error in the asymptotic limit then the corresponding error exponent is famously given by the Chernoff exponent~\cite{ACM+07, ANS+08}.  

A different approach is to instead minimize the number of samples needed to achieve a certain error probability. This is called the sample complexity of the discrimination problem and is defined as 
\begin{align} \label{eq:defn:sym}
    \SC_{\sym}(\rho,\sigma,p,\delta) := \inf\{ n \,|\, p_{e}(\rho^{\otimes n},\sigma^{\otimes n},p) \leq\delta \}. 
\end{align}
Here, the subscript `B' stands for the Bayesian setting.
For simplicity, we will sometimes fix $p=0.5$ and $\delta=0.1$ and then drop them from the notation. It is a folklore result that the sample complexity in the setting with equal priors is given by \cite{FG99, Wilde2024_review}
\begin{align}
    \SC_{\sym}\left(\rho,\sigma\right) = \Theta\left(\frac1{-\log F(\rho,\sigma)}\right),
\end{align}
where $F(\rho,\sigma)$ is the quantum fidelity.
A related problem concerned with an intermediate setting is that of sequential hypothesis testing~\cite{vargas2021quantum,li2022optimal}.
In this work, we consider the sample complexity when the state in question is affected by a noisy channel before we receive it. Explicitly, we are considering the class of $\epsilon$-locally differentially private quantum channels, denoted $\LDP_\epsilon$, which is defined by~\cite{hirche2023quantumDP}
\begin{align}
    \LDP_\epsilon = \{ \cA\,|\, E_{e^\epsilon}(\cA(\rho)\|\cA(\sigma))=0\quad \forall\rho,\sigma \}. 
\end{align}
That is, we are interested in the sample complexity
\begin{align}
  \SC_{\sym,\epsilon}(\rho,\sigma,p,\delta) := \inf_{\cA\in \LDP_\epsilon}\,\SC_{\sym}(\mathcal{A}(\rho),\mathcal{A}(\sigma),p,\delta)\,, 
\end{align}
namely the smallest possible sample size given that the states are subject to locally differetially private noise. In the classical case such problems have received a fair amount of attention recently~\cite{Evfimievski_2003, Kasiviswanathan_2008, duchi2013local, kamath2020LDP}. In the case of equal priors, $p=0.5$, the classical equivalent of our particular problem was investigated in~\cite{asoodeh2022contraction}. There, lower and upper bounds on differentially private sample complexity were given. In this work, we build on these results to investigate local differential privacy in the sample complexity of quantum state discrimination. 
The main result of our work can be summarized by the following inequalities. For equal priors, we find, 
\begin{align}
    &\max\left\{\frac{(e^\epsilon+1)\log 2.5}{2(e^\epsilon-1)H_{\frac12}(\rho\|\sigma)}, \frac{16}{25} \frac{e^{\epsilon}}{(e^\epsilon-1)^2\,E_1(\rho\|\sigma)^2 } \right\} \leq \SC_{\sym,\epsilon}\left(\rho,\sigma,0.5,0.1\right) \leq \left(\frac{e^\epsilon+1}{e^\epsilon-1}\right)^2\frac{2 \log 5}{ E_1(\rho\|\sigma)^2}. 
\end{align}
These mirror mostly the classical result in~\cite[Lemma 2]{asoodeh2022contraction}, up to some differences in the $\epsilon$-dependent constants that we will discuss later. 
The sample complexity with arbitrary priors was recently investigated in~\cite{pensia2024sample,Wilde2024_review}. We extend the discussion of local differential privacy to this setting and find, 
\begin{align} \label{eq:bounds_p/4}
        \max\left\{ \frac3{16}\frac{e^\epsilon+1}{e^\epsilon-1} \frac{p\log{p^{-1}}}{\JS_p(\rho\|\sigma)}, \frac9{16}\frac{1}{e^{-\epsilon}(e^\epsilon-1)E_1^2(\rho\|\sigma)}  \right\} \leq \SC_{\sym,\epsilon }\left(\rho,\sigma,p,\frac{p}4\right) &\leq 
         \left(\frac{e^\epsilon+1}{e^\epsilon-1}\right)^2 \frac{2\log{p^{-1}}+2\log{4}}{ E_1(\rho\|\sigma)^2}\,.
\end{align}
It should be noted that the optimal bounds on the sample complexity necessarily depend on the relationship between $\delta$ and $p$. This is why we choose $\delta=\frac{p}4$, which gives the most interesting region. For a detailed discussion we refer to~\cite{pensia2024sample} and the later sections. 
In particular, we recover the interesting observation that the best known lower bound is given by a different divergence depending on the value of $\epsilon$. Note that this ``phase-transition'' is known to persist in the exact behaviour of binary classical probability distributions~\cite{pensia2023simple}. 

The main challenge in deriving the above result lies in the need of several new entropic inequalities that we believe should also be of independent interest. To that end we use the recently established framework of $f$-divergences defined via integral representations from~\cite{hirche2023quantum}. In particular, these divergences behave well when we consider their contraction coefficients which allows us to prove new bounds for several such coefficients under local differential privacy. 

Lastly, the above established bounds along with the equivalence \cite[Corollary 4.8]{pensia2024sample} allow us to obtain bounds on sample complexity for asymmetric hypothesis testing as follows: if the type-I error $\alpha$ and type-II error $\beta$ are at most $1/8$, then
\begin{align}
\frac{1}{e^{-\epsilon}(e^\epsilon-1)E_1(\rho\|\sigma)^2}
        \lesssim \SC_{\asym,\epsilon }\left(\rho,\sigma,\alpha,\beta\right) &\lesssim
         \left(\frac{e^\epsilon+1}{e^\epsilon-1}\right)^2 \frac{1 }{ E_1(\rho\|\sigma)^2},
\end{align}
where `$\lesssim$' means an inequality up to a universal constant and we refer the readers to Section~\ref{sec:asym} for more precise definitions of $\SC_{\asym,\epsilon }\left(\rho,\sigma,\alpha,\beta\right)$.

\textit{Note:} During the completion of this work we became aware of the related work~\cite{Theshani2024} that also considers hypothesis testing under differential privacy constraints. 

\section{Preliminaries}
\subsection{Notations}
We denote by $\mathcal{S}_d$ the set of $d$-dimensional quantum states, and by $\mathcal{L}_d$ that of $d\times d$ arbitrary matrices. We will also adopt standard notations from quantum information theory: for a quantum system $A$, $|A|$ denotes the dimension of its associated Hilbert space, $\mathcal{S}_A $ the set of its quantum states, and $\mathcal{L}_A\equiv \mathcal{L}_{|A|}$. Given a self-adjoint matrix $H$, we denote by $\{H\ge 0\}$ the projection onto the sum of eigensubspaces corresponding to the non-negative eigenvalues of $H$, and write $H_+:=H\{H\ge 0\}$. Similarly, for two states $\rho,\sigma\in\mathcal{S}_A$, $\{\rho\ge \sigma\}:=\{\rho-\sigma\ge 0\}$.  
We use the natural logarithm $\log$ throughout this paper.
We use $I$ for the identity matrix.

\subsection{Divergences and distance measures}

The classical Renyi divergence is defined as, 
\begin{align*}
D_\alpha(P\|Q)&= \frac{1}{\alpha-1} \log\left( 1 +(\alpha-1)H_{\alpha}(P\|Q)\right) =\frac{1}{\alpha-1}\log(\tr(P^\alpha Q^{1-\alpha})) 
\end{align*}
where,
\begin{align}
    H_\alpha(P\|Q)=\frac{1}{\alpha-1}\big(\tr(P^\alpha Q^{1-\alpha})-1\big)\,.
\end{align}
is the Hellinger divergence.
An important special case which we will consider here is the case $\alpha=1/2$. Then $H_{\frac{1}{2}}(P,Q)=2\big(1-\tr P^{\frac{1}{2}} Q^{\frac{1}{2}}\big)=\tr\big[(\sqrt{P}-\sqrt{Q})^2\big]$ corresponds to (twice) the squared Hellinger distance between $P$ and $Q$. 

In the quantum setting, we have several different definitions of \Renyi divergences. Recently, a family of \Renyi divergences was introduced based on the following integral representation~\cite{hirche2023quantum}: for $\alpha>0$ with $\alpha\ne 1$,
\begin{align}\label{Eq:Def-Da-Ha}
    D_{\alpha}(\rho\|\sigma) = \frac{1}{\alpha-1} \log\left( 1 +(\alpha-1)H_{\alpha}(\rho\|\sigma)\right),
\end{align}
with
\begin{align}
    \label{eq:H_alpha}
    H_{\alpha}(\rho\|\sigma) = \alpha\int_1^\infty\Big(\gamma^{\alpha-2}E_\gamma(\rho\|\sigma)+\gamma^{-\alpha-1}E_\gamma(\sigma\|\rho)\,\Big)d\gamma\,, 
\end{align}
where $E_\gamma(\rho\|\sigma):=\tr(\rho-\gamma\sigma)_+$ denotes the quantum Hockey-Stick divergence. Note that more generally,~\cite{hirche2023quantum} defines a family of quantum $f$-divergences,
\begin{align}
    \label{eq:D_f}
    D_{f}(\rho\|\sigma) = \int_1^\infty\Big(f''(\gamma) E_\gamma(\rho\|\sigma)+\gamma^{-3} f''(\gamma^{-1}) E_\gamma(\sigma\|\rho)\,\Big)d\gamma\,, 
\end{align} which gives the Hellinger divergence for $f(x)=\frac{x^\alpha-1}{\alpha-1}$. 
Two alternative definitions of \Renyi divergences are the Petz \Renyi divergence~\cite{petz1986quasi}, 
\begin{align}
    \widebar D_\alpha(\rho\|\sigma) &:= \frac{1}{\alpha -1} \log \tr \left( \rho^{\alpha}\sigma^{1-\alpha} \right) , \qquad &&\textnormal{for}\ \alpha \in (0, 1) \cup (1, 2] ,
\end{align}
and the sandwiched \Renyi divergence~\cite{muller2013quantum,wilde2014strong},
\begin{align}
    \widetilde D_\alpha(\rho\|\sigma) &= \frac{1}{\alpha -1} \log \tr \left( \left(\sigma^{\frac{1-\alpha}{2\alpha}}\rho\sigma^{\frac{1-\alpha}{2\alpha}}\right)^{\alpha} \right), \qquad &&\textnormal{for}\ \alpha \in \Big[\frac12, 1\Big) \cup (1, \infty) .
\end{align}
We define the corresponding Hellinger divergences $\widebar H_\alpha(\rho\|\sigma)$ and $\widetilde H_\alpha(\rho\|\sigma)$ analog to Equation~\eqref{Eq:Def-Da-Ha}. 

We recall that the fidelity is given by the $\frac{1}{2}$-sandwiched \Renyi divergence via $\widetilde D_{\frac12}(\rho\|\sigma) = -2\log F(\rho,\sigma)$,
where $F(\rho,\sigma) := \|\sqrt{\rho}\sqrt{\sigma}\|_1$ denotes the fidelity between the quantum states $\rho$ and $\sigma$,
and $d_B(\rho,\sigma) := \sqrt{1-F(\rho,\sigma)}$ is their so-called Bures distance.  In the following, we will also make use of the max-relative entropy $D_{\infty}(\rho\|\sigma):=\log\|\sigma^{-\frac{1}{2}}\rho\sigma^{-\frac{1}{2}}\|_\infty$, where we assume that $\sigma$ is invertible for simplicity.
Note that all the R\'{e}nyi divergences defined above are different in general, but they coincide when $\rho$ and $\sigma$ commute and correspond to two probability distributions $P$ and $Q$. 
We also recall that for all of them the limit $\alpha\to 1$ leads to the standard Umegaki relative entropy $D_1(\rho\|\sigma)\equiv D(\rho\|\sigma):=\tr\big[\rho(\log\rho-\log\sigma)\big]$.
The Jensen-Shannon divergence is defined as
\begin{align}
    \JS_p(\rho\|\sigma) = p D(\rho\|p\rho+(1-p)\sigma) + (1-p) D(\sigma\|p\rho+(1-p)\sigma). 
\end{align}
Finally we remark that the aforementioned $f$-divergences from~\cite{hirche2023quantum} also include the relative entropy for $f(x)=x\log x$ and the Jensen-Shannon divergence for $f(x)=px\log(x) + (1-p+px)\log((1-p+px)^{-1})$. 

\section{Sample complexity of hypothesis testing}

Given any two states $\rho,\sigma\in \mathcal{S}_d$, the optimal sample complexity for the task of hypothesis testing between $\rho$ and $\sigma$, namely the number of copies needed to distinguish $\rho$ from $\sigma$, satisfies \cite{Wilde2024_review}
\begin{align}\label{SCquantum}
    \SC_{\sym}(\rho,\sigma) \!=
    \Theta\left(\frac{-1}{\log F(\rho,\sigma)}\right)  \!= \Theta\left(\frac1{d_B^2(\rho,\sigma)}\right)\!.
\end{align}
It is a well-known fact that the sample complexity of hypothesis testing between two distributions satisfies \begin{align}
    \SC_{\sym}(P,Q)=\Theta\left(\frac{1}{H_{\frac{1}{2}}(P,Q)}\right)
\end{align} 
(see \cite{bar2002complexity} for the lower bound, and \cite{canonne2019structure} for the upper bound).
To recover this classical asymptotic relation from \eqref{SCquantum}, we need to relate $H_{\frac12}(\rho\|\sigma)$ to the quantum fidelity.
\begin{lemma}\label{Lem:FvdG}
   For any two states $\rho,\sigma\in\mathcal{S}_d$,
    \begin{align}
        1 - F(\rho,\sigma) \leq \frac{1}{2} H_{\frac12}(\rho\|\sigma) \leq E_1(\rho\|\sigma) \leq \sqrt{1-F(\rho,\sigma)^2} \leq \sqrt{H_{\frac12}(\rho\|\sigma)}.
    \end{align}
\end{lemma} 
\begin{proof}
    The first  
 inequality follows from $\widetilde D_{\frac12}(\rho\|\sigma) \leq D_{\frac12}(\rho\|\sigma)$ and the second from $E_\gamma \leq E_1$. This is essentially also the argument in the proof of~\cite[Corollary 5.6]{hirche2023quantum}. The third inequality is the usual Fuchs-van-de-Graaf inequality \cite{FG99}. The final inequality follows from the first. 
\end{proof}
Hence, we can argue from these bounds that also in the quantum setting we can at least partially express SC in terms of $H_{\frac12}(\rho\|\sigma)$. Note that the following fidelity based bounds are similar to results recently provided in~\cite{Wilde2024_review}. We provide a full proof for completeness. 
\begin{proposition} \label{prop:SC_original}
For $\delta\leq p\leq\frac12$ and any two states $\rho$ and $\sigma$ with $H_{\frac12}(\rho\|\sigma) \leq 1$, we have
\begin{align}\label{asymptquant2}
    \frac{\frac12 \log{\frac{p(1-p)}{\delta(1-\delta)}}}{   H_{\frac12}(\rho\|\sigma)  }
    \leq \frac{\frac12 \log{\frac{p(1-p)}{\delta(1-\delta)}}}{ - \log F(\rho,\sigma) }
    \leq
    \SC_{\sym}\left(\rho,\sigma,p, \delta\right)
&\leq
\frac{ \log (\frac{1-p}{\delta}) }{ -\log F(\rho,\sigma) }
\leq
\frac{ \log (\frac{1-p}{\delta}) }{ 1 - F(\rho,\sigma) }\,.
\end{align}
Specifically, for $p=\frac12$ and $\delta\leq\frac14$, we have
\begin{align}\label{12asymptquant2}
    \frac{\frac12 \log{\frac{1}{4\delta}}}{   H_{\frac12}(\rho\|\sigma)  }
    \leq \frac{\frac12 \log{\frac{1}{4\delta}}}{ - \log F(\rho,\sigma) }
    \leq
    \SC_{\sym}\left(\rho,\sigma,0.5, \delta\right)
&\leq
\frac{ \log (\frac{1}{2\delta}) }{ -\log F(\rho,\sigma) }
\leq
\frac{ \log (\frac{1}{2\delta}) }{ 1 - F(\rho,\sigma) }\,.
\end{align}
\end{proposition}
\begin{proof}
We first show the upper bound.
By Equation~\eqref{Eq:Error-prob} and Lemma~\ref{Lem:FvdG},
\begin{align}
\delta 
&= p - p E_\frac{1-p}{p}(\rho^{\otimes n}\|\sigma^{\otimes n})
\\ 
&\leq p - p \left( 1- \frac{1-p}{p} E_1(\rho^{\otimes n}\|\sigma^{\otimes n}) \right) \\
&= (1-p)\left( 1- E_1(\rho^{\otimes n}\|\sigma^{\otimes n}) \right) \\
&\leq (1-p) F\left( \rho^{\otimes n} \| \sigma^{\otimes n} \right)
\\
&= (1-p) F\left( \rho \| \sigma \right)^n,
\end{align}
where the first inequality is from~\cite[Lemma II.4]{hirche2023quantumDP}. We obtain
\begin{align}
\SC_{\sym}\left(\rho,\sigma,p, \delta\right)
&\leq
\frac{ \log (\frac{1-p}{\delta}) }{ -\log F(\rho,\sigma) }
\leq
\frac{ \log (\frac{1-p}{\delta}) }{ 1 - F(\rho,\sigma) },
\end{align}
where the last inequality follows from 
$\log(1+x)\leq x$ for all $x > -1$.

Next, we show the lower bound. By~\cite[Lemma II.3]{hirche2023quantumDP}, we have, 
\begin{align}
    E_\gamma(\rho\|\sigma) \leq \frac12 \sqrt{(1+\gamma)^2-4\gamma F(\rho,\sigma)^2}. 
\end{align}
Following the idea in \cite[Theorem 4.7]{bar2002complexity} and using Lemma~\ref{Lem:FvdG}, we then have
\begin{align}
 F(\rho,\sigma)^n
 &= F\left(\rho^{\otimes n},\sigma^{\otimes n}\right) 
 \\
 &\leq \sqrt{\frac{1}{p(1-p)}\left[ \frac{1}{4} - \left(\frac12 -p +p E_{\frac{1-p}{p}}\left(\rho^{\otimes n},\sigma^{\otimes n}\right)\right)^2 \right]}
 \\
 &\leq \sqrt{\frac{1}{p(1-p)}\left[ \frac{1}{4} - \left(\frac12 -\delta\right)^2 \right]}
 \\
 &= \sqrt{\frac{\delta(1-\delta)}{p(1-p)} }.
\end{align}
This gives,
\begin{align}
    n \geq \frac{\frac12 \log{\frac{p(1-p)}{\delta(1-\delta)}}}{-\log F(\rho,\sigma)}. 
\end{align}
Finally, we bound this in terms of $H_{\frac12}(\rho\|\sigma)$, 
\begin{align}
\SC_{\sym}\left(\rho,\sigma,0.5,\delta\right)
&\geq
\frac{\frac12 \log{\frac{p(1-p)}{\delta(1-\delta)}}}{ - \log F(\rho,\sigma) }
\\
&\geq
\frac{\frac12 \log{\frac{p(1-p)}{\delta(1-\delta)}}}{  \log \left( 1 - \frac{H_{\frac12}(\rho\|\sigma)}{2} \right) }
\\
&\geq \frac{\frac12 \log{\frac{p(1-p)}{\delta(1-\delta)}}}{   H_{\frac12}(\rho\|\sigma)  },
\end{align}
where the last inequality follows from $\log (1-x) \geq -2x$, for all $0\leq x \leq \frac{1}{2}$ provided that 
$H_{\frac12}(\rho\|\sigma)  \leq 1$.
\end{proof}
As was pointed out in~\cite{pensia2024sample}, the general lower and upper bounds provided in this way don't match in all parameters. In particular, the bounds in Equation~\eqref{asymptquant2} are tight when $\delta$ depends quadratic on $p$, specifically $\delta\leq\frac{p^2}{4}$. The counterpoint is set by linear dependence, here $\delta=\frac{p}4$, for which we derive new bounds below. These are quantum generalizations of~\cite[Theorem 2.1]{pensia2024sample}. 
\begin{theorem}\label{Thm:SC-ap}
For $\delta=\frac{p}4$ and $p\in(0,\frac12]$, 
    \begin{align}
        \frac{3\sqrt{2}}{16}  \frac{\log2}{\lambda  H_{1-\lambda}(\rho\|\sigma)} \leq \frac3{16} \frac{p\log{p^{-1}}}{\JS_p(\rho\|\sigma)}\leq  \SC_{\sym}(\rho,\sigma,p,\frac{p}4) \leq \frac{2}{\lambda \widebar H_{1-\lambda}(\rho\|\sigma)}, 
    \end{align}
    where $\lambda=\frac{\log2}{2\log{p^{-1}}}$. 
\end{theorem}
\begin{proof}
    The first inequality follows from Corollary~\ref{Cor:L74}, the second is Proposition~\ref{Prop:prior-lower} and the upper bound is Proposition~\ref{Prop:prior-upper}. 
\end{proof}
The individual technical contributions will be proven in the following section. There we also include more general relationships between $p$ and $\delta$. Note that the upper bound uses the Petz Hellinger divergence, while the lower bound usues the Hellinger divergence defined in~\cite{hirche2023quantum}. The precise relation between the two remains an open problem, although there are examples in~\cite{hirche2023quantum} for which $\widebar H_s(\rho\|\sigma)\leq H_s(\rho\|\sigma)$. 

\subsection{Proof of Theorem~\ref{Thm:SC-ap}}

We start with the following observation, giving two properties of the Jensen-Shannon divergence.  
\begin{lemma}\label{lem:JS-prop}
    The following two properties hold. 
    \begin{enumerate}
        \item We have, 
    \begin{align}
        \JS_p(\rho\|\sigma) = I(\Theta : A),
    \end{align}
    where $\tau_{\Theta A}=p |0\rangle\langle 0|\otimes \rho + (1-p) |1\rangle\langle 1|\otimes\sigma$.
    \item The Jensen-Shannon divergence is subadditive, i.e. 
    \begin{align}
        \JS_p(\rho^{\otimes n}\|\sigma^{\otimes n}) \leq n \JS_p(\rho\|\sigma).
    \end{align}
    \end{enumerate}
\end{lemma}
\begin{proof}
    We start by proving the first claim. We note,
    \begin{align}
        I(\Theta : A) = D(\tau_{\Theta A}\|\tau_{\Theta }\otimes \tau_{A}) 
        = p D(\rho\|\tau_{A}) + (1-p) D(\sigma\|\tau_{A}) = \JS_p(\rho\|\sigma),  
    \end{align}
    where $\tau_{A}=\tr_\Theta \tau_{\Theta A}=p \rho + (1-p)\sigma$. 
    To prove the second claim, set $\tau_{\Theta A^n}=p |0\rangle\langle 0|\otimes \rho^{\otimes n} + (1-p) |1\rangle\langle 1|\otimes\sigma^{\otimes n}$. Then, using the first claim, we have,
    \begin{align}
        \JS_p(\rho^{\otimes n}\|\sigma^{\otimes n}) &= I(\Theta : A^n) \\
        &= H(A^n) - H(A^n|\Theta) \\
        &\leq n H(A) - H(A^n|\Theta) \\
        &= n H(A) - p H(\rho^{\otimes n})- (1-p) H(\sigma^{\otimes n}) \\
        &= n H(A) - n p H(\rho)- n (1-p) H(\sigma) \\
        &= n I(\Theta:A)=n \JS_p(\rho\|\sigma),
    \end{align}
    where the inequality is subadditivity of the entropy. This completes the proof. 
\end{proof}
With this, we can give a lower bound on the sample complexity. 
\begin{proposition}\label{Prop:prior-lower}
    For $p\in(0,\frac12]$ and $p_e^*\leq p(1-\gamma)$, we have 
    \begin{align}
        \SC_{\sym}\left(\rho,\sigma,p,p(1-\gamma)\right)  \geq \frac{p\gamma\log{\frac{1-p}{p}}+p^2\gamma^2}{\JS_p(\rho\|\sigma)}
    \end{align}
    and in particular for $\gamma=\frac34$, 
    \begin{align}
        \SC_{\sym}\left(\rho,\sigma,p,\frac{p}{4}\right) 
 \geq \frac3{16} \frac{p\log{p^{-1}}}{\JS_p(\rho\|\sigma)}. 
    \end{align}
\end{proposition}
\begin{proof}
    We begin by noting, 
    \begin{align}
        I(\Theta : A^n) &= \JS_p(\rho^{\otimes n}\|\sigma^{\otimes n}) \leq n \JS_p(\rho\|\sigma), 
    \end{align}
    which follows from Lemma~\ref{lem:JS-prop}. This implies
    \begin{align}
        H(\Theta|A^n) \geq H(\Theta) + n \JS_p(\rho\|\sigma) \\
        = h(p) - n \JS_p(\rho\|\sigma). 
    \end{align}
    Then by Fano's inequality,
    \begin{align}
        h(p_e^*) \geq  h(p) - n \JS_p(\rho\|\sigma). 
    \end{align}
    Now, fixing $p\in(0,\frac12]$ and $p_e^*\leq p(1-\gamma)$ and following the exact same argument as~\cite{pensia2024sample}, we get
    \begin{align}
        n \geq \frac{p\gamma\log{\frac{1-p}{p}}+p^2\gamma^2}{\JS_p(\rho\|\sigma)}.
    \end{align}
    Further, specializing the above to $\gamma=\frac34$, one finds
    \begin{align}
        \frac3{16} \frac{p\log{p^{-1}}}{\JS_p(\rho\|\sigma)}. 
    \end{align}
    This concludes the proof. 
\end{proof}

For the next result, we give an upper bound on the sample complexity in terms of the Petz Hellinger divergence. 
\begin{proposition}\label{Prop:prior-upper}
    We have, 
    \begin{align}
        \SC_{\sym}(\rho,\sigma,p,\delta) &\leq \frac{\log{\left(p^{1-\lambda}(1-p)^\lambda/\delta\right)}}{\lambda \widebar H_{1-\lambda}(\rho\|\sigma)} .
    \end{align}
    And if we fix $\delta=\frac{p}4$ and $\lambda=\frac{\log2}{2\log{p^{-1}}}$, then we have 
    \begin{align}
        \SC_{\sym}\left(\rho,\sigma,p,\frac{p}{4}\right) &\leq \frac{2}{\lambda \widebar H_{1-\lambda}(\rho\|\sigma)}. 
    \end{align}
\end{proposition}
\begin{proof}
    We have, compare~\cite[Corollary 3.5]{hirche2023quantum},
    \begin{align}
        E_\gamma(\rho\|\sigma) &\geq 1-\gamma^{1-s}\widebar Q_s(\rho\|\sigma).
    \end{align}
    Thus we can bound, using $\lambda=1-s$,  
    \begin{align}
        p_e^*(\rho^{\otimes n},\sigma^{\otimes n},p) & = p - p E_\frac{1-p}{p}(\rho^{\otimes n}\|\sigma^{\otimes n}) \\
        &\leq p - p \left(1-\left(\frac{1-p}{p}\right)^{\lambda}\widebar Q_{1-\lambda}(\rho^{\otimes n}\|\sigma^{\otimes n})\right) \\
        &= p^{1-\lambda} (1-p)^{\lambda} \widebar Q_{1-\lambda}(\rho\|\sigma)^n \\
        &= p^{1-\lambda} (1-p)^{\lambda} (1-\lambda \widebar H_{1-\lambda}(\rho\|\sigma))^n \\
        &\leq p^{1-\lambda} (1-p)^{\lambda} e^{-n\lambda \widebar H_{1-\lambda}(\rho\|\sigma)},
    \end{align}
    where we used $1-x\leq e^{-x}$ in the final inequality. This implies that we get $p_e^*(p,\rho^{\otimes n},\sigma^{\otimes n})\leq\delta$, if
    \begin{align}
        n\geq \frac{\log{\left(p^{1-\lambda}(1-p)^\lambda/\delta\right)}}{\lambda \widebar H_{1-\lambda}(\rho\|\sigma)}.
    \end{align}
    Hence, 
    \begin{align}
        \SC_{\sym}(\rho,\sigma,p,\delta) &\leq \frac{\log{\left(p^{1-\lambda}(1-p)^\lambda/\delta\right)}}{\lambda \widebar H_{1-\lambda}(\rho\|\sigma)} \\
        &= \frac{\log{\left(p^{-\lambda}(1-p)^\lambda/(1-\gamma)\right)}}{\lambda \widebar H_{1-\lambda}(\rho\|\sigma)} \\
        &\leq \frac{\lambda\log{\left(p^{-1}\right)}-\log{(1-\gamma)}}{\lambda \widebar H_{1-\lambda}(\rho\|\sigma)}, 
    \end{align}
    where we fixed $\delta=p(1-\gamma)$ for the equality. Finally, fix $\lambda=\frac{\log2}{2\log{p^{-1}}}$ and $\gamma=\frac34$, then
    \begin{align}
        \SC_{\sym}(\rho,\sigma,p,\frac{p}4) &\leq \frac{\lambda\log{\left(p^{-1}\right)}-\log{(1-\gamma)}}{\lambda \widebar H_{1-\lambda}(\rho\|\sigma)} \\
        &= \frac{\frac12\log{2}-\log{\frac14}}{\lambda \widebar H_{1-\lambda}(\rho\|\sigma)} \\
        &\leq \frac{2}{\lambda \widebar H_{1-\lambda}(\rho\|\sigma)}. 
    \end{align}
    This completes the proof. 
\end{proof}

In~\cite{pensia2024sample} several bounds on the sample complexity of hypothesis testing are proven. In that work, an important role is taken by their~\cite[Lemma 7.4]{pensia2024sample}, which essentially states, 
\begin{align}
    \JS_\alpha(p\|q) \leq \alpha 32 e^{2\lambda\log(\frac1\alpha)} H_{1-\lambda}(p\|q), 
\end{align}
for $\alpha\in(0,\frac12]$ and $\lambda\in(0,\frac12]$.

We now proof the following lemma that is much tighter and also holds for quantum states. 
\begin{lemma} For any $\alpha\in(0,1)$ and $\lambda\in(0,1)$, 
    \begin{align}
    \JS_\alpha(\rho\|\sigma) \leq \alpha \left(\frac{\lambda(1-\alpha)}{(1-\lambda)\alpha}\right)^{\lambda} H_{1-\lambda}(\rho\|\sigma), 
\end{align}
\end{lemma}
\begin{proof}
    We want to prove something of the form of 
    \begin{align}
        D_f(\rho\|\sigma) \leq c \lambda D_g(\rho\|\sigma).
    \end{align}
    Using the integral representation of the f-divergences, this holds if 
    \begin{align}
        f''(\gamma) \leq c \lambda g''(\gamma),\quad\forall \gamma\in[0,\infty]. 
    \end{align}
    For the Jensen-Shannon divergence $JS_\alpha$, we have
    \begin{align}
        f_\alpha''(x)=\frac{\alpha(1-\alpha)}{x(1-\alpha+\alpha x)}, 
    \end{align}
    and for the Hellinger divergence $H_{1-\lambda}$, 
    \begin{align}
        g_\lambda''(x)=(1-\lambda)x^{-\lambda-1}.
    \end{align}
    Hence, we need to find a $c$ such that 
    \begin{align}
        &\frac{\alpha(1-\alpha)}{x(1-\alpha+\alpha x)} \leq c \lambda  (1-\lambda)x^{-\lambda-1} \\
        \Leftrightarrow &\frac{\alpha(1-\alpha)}{\lambda(1-\lambda)} \leq c   (1-\alpha+\alpha x)x^{-\lambda} =: c\cdot h(x). \label{Eq:hx}
        \end{align}
    We are now left with finding the minimum of 
    \begin{align}
        h(x) = \frac{1-\alpha}{x^\lambda} + \alpha x^{1-\lambda}. 
    \end{align}
    This function diverges to $+\infty$ for $x\in\{0,\infty\}$ and has a unique minimum in between. To find the minimum, we check
    \begin{align}
        h'(x) = -\lambda\frac{1-\alpha}{x^{\lambda+1}} + (1-\lambda)\alpha x^{-\lambda}.  
    \end{align}
    Setting $h'(x_0)=0$ gives, 
    \begin{align}
        &\lambda\frac{1-\alpha}{x_0}=   (1-\lambda)\alpha \\
        \Leftrightarrow\quad &x_0 = \frac{\lambda(1-\alpha)}{(1-\lambda)\alpha}.
    \end{align}
    Putting this back into Equation~\eqref{Eq:hx}, we have
    \begin{align}
        h(x_0) =  (1-\alpha+\alpha x_0)x_0^{-\lambda} \\
        =  (1-\alpha+\alpha \frac{\lambda(1-\alpha)}{(1-\lambda)\alpha})\left(\frac{\lambda(1-\alpha)}{(1-\lambda)\alpha}\right)^{-\lambda} \\
        =  \left(\frac{1-\alpha}{1-\lambda}\right)\left(\frac{\lambda(1-\alpha)}{(1-\lambda)\alpha}\right)^{-\lambda}.
    \end{align}
    Hence we can choose any $c$ such that 
    \begin{align}
        &\frac{\alpha(1-\alpha)}{\lambda(1-\lambda)} \leq c   \left(\frac{1-\alpha}{1-\lambda}\right)\left(\frac{\lambda(1-\alpha)}{(1-\lambda)\alpha}\right)^{-\lambda} \\
        \Leftrightarrow\quad & \frac{\alpha}{\lambda} \left(\frac{\lambda(1-\alpha)}{(1-\lambda)\alpha}\right)^{\lambda} \leq c 
    \end{align}
    This gives the claimed statement. 
\end{proof}
This might look quite different from the previous result, but we can rewrite it as 
    \begin{align}
         \JS_\alpha(\rho\|\sigma) \leq \alpha \left(\frac{\lambda(1-\alpha)}{(1-\lambda)}\right)^{\lambda} e^{\lambda\log(\frac1\alpha)} H_{1-\lambda}(\rho\|\sigma). 
    \end{align}
For the special case considered in~\cite{pensia2024sample}, we can simplify this as follows. 
\begin{corollary}
    For any $\alpha\in(0,1)$ and $\lambda\in(0,\frac12]$, 
    \begin{align}
    \JS_\alpha(\rho\|\sigma) &\leq \alpha (1-\alpha)^{\lambda} e^{\lambda\log(\frac1\alpha)} H_{1-\lambda}(\rho\|\sigma) \\
    &\leq \alpha e^{\lambda\log(\frac1\alpha)} H_{1-\lambda}(\rho\|\sigma)
\end{align}
\end{corollary}
\begin{proof}
    This follows simply from
    \begin{align}
        \left(\frac{\lambda}{1-\lambda}\right)^{\lambda} \leq \left(\frac{1/2}{1-1/2}\right)^{\lambda} = 1
    \end{align}
\end{proof}
Hence we are improving on the previous result by at least a factor $32$. 
Finally, we can apply this to the special case considered in the second half of~\cite[Lemma 7.4]{pensia2024sample}. 
\begin{corollary}\label{Cor:L74}
    For any $\alpha\in(0,\frac12]$ and $\lambda=\frac{\log2}{2\log\frac1\alpha}$, we have
    \begin{align}
        \JS_\alpha(\rho\|\sigma) \leq \alpha \sqrt{2} H_{1-\lambda}(\rho\|\sigma)
    \end{align}
\end{corollary}
\begin{proof}
    Note, 
    \begin{align}
        e^{\lambda\log(\frac1\alpha)} = e^{\frac{\log2}{2\log\frac1\alpha}\log(\frac1\alpha)}=\sqrt{2}. 
    \end{align}
\end{proof}
This implies for $\lambda=\frac{\log2}{2\log\frac1\alpha}$, 
\begin{align}
    \ceil*{\frac{2}{\lambda H_{1-\lambda}(\rho\|\sigma)}} \leq \ceil*{\frac{4\sqrt{2}\,\alpha\log\frac1\alpha}{\log2 \, \JS_\alpha(\rho\|\sigma)}},  
\end{align}
which can be compared to the previous result~\cite{pensia2024sample}, 
\begin{align}
    \ceil*{\frac{2}{\lambda H_{1-\lambda}(p\|q)}} \leq \ceil*{\frac{256\,\alpha\log\frac1\alpha}{\log2 \, \JS_\alpha(p\|q)}},  
\end{align}
and hence we got an improvement of $32\sqrt{2}$.

\section{Locally differentially private hypothesis testing}

This paper aims at finding tight bounds on the sample complexity of optimal hypothesis tests subject to local privacy guarantees \cite{evfimievski2003limiting,kasiviswanathan2011can}. Loosely speaking, a random mechanism (e.g. an algorithm or communication protocol) is said to be locally differentially private (LDP) if its output does not vary significantly with arbitrary perturbation of the input. LDP was initially introduced in the classical setting for the scenario where a database is compiled from numerous clients, each insisting on individual privacy assurances. In this scenario, each client employs an algorithm $\mathcal{A}$ to obfuscate their input to the database. The primary objective is not to create similarities between neighboring states, but rather to conceal the broader information being transmitted.

In the present article, we focus on the task of LDP hypothesis testing. There, identical copies of a state $\omega\in\{\rho,\sigma\}$ are given and we need to discriminate between the hypotheses $\omega=\rho$ and $\omega=\sigma$. However, as opposed to the nonprivate setting, we are restricted in the set of positive operator-valued measures (POVM) which we can use in order to complete the task:

\begin{definition}[$\epsilon$-locally differentially private channel \cite{hirche2023quantumDP}] Given $\epsilon\ge 0$, a quantum channel $\mathcal{A}:\mathcal{L}_A\to \mathcal{L}_B$ is called $\epsilon$-locally differentially private if for all states $\rho,\sigma\in\mathcal{S}_A$,
\begin{align*}
E_{e^\epsilon}(\mathcal{A}(\rho)\|\mathcal{A}(\sigma))=0\,. 
\end{align*}
We denote by $\LDP_\epsilon(A,B)$ the set of $\epsilon$-locally differentially private quantum channels from $A$ to $B$.
\end{definition}
Then, given two states $\rho,\sigma\in\mathcal{S}_A$ and $\epsilon\ge  0$, their optimal sample complexity for $\epsilon$-LDP hypothesis testing is defined as
\begin{align}
  \SC_{\sym,\epsilon}(\rho,\sigma,p,\delta) := \inf_{\cA\in \LDP_\epsilon(A,B)}\,\SC_{\sym}(\mathcal{A}(\rho),\mathcal{A}(\sigma))\,. 
\end{align}
Here, we often will not specify the output system to the locally differentially private algorithms $\mathcal{A}$ as it won't play an important role in our derivations. 
Again we will throw away the dependence of $p$ and $\delta$ if they are some fixed constants.
From \eqref{asymptquant2}, we directly get that 
\begin{align}
    \SC_{\sym,\epsilon}(\rho,\sigma) = \Theta\left(\frac{-1}{\sup_{\cA\in \LDP_\epsilon} \log F(\cA(\rho)\|\cA(\sigma))}\right).
\end{align}
The above expression is somewhat unsatisfactory due to the presence of an optimization over all LDP mechanisms. Ideally, we would like to derive tight upper and lower bounds for the $\SC_{\sym,\epsilon}$ which do not depend on such optimization.

In the classical setting, fundamental limits of statistical problems under LDP have been successfully characterized using information-theoretic concepts. Arguably one of the most fundamental notions in (quantum) information theory revolves around data processing. Under the influence of a quantum channel, numerous relevant quantities exhibit monotonic behavior. This characteristic allows us to attribute operational significance to these quantities concerning distinguishability, consequently facilitating their utility in assessing physical properties. For instance, the data processing inequality states that, for any two states $\rho,\sigma\in\mathcal{S}_A$ and any quantum channel $\mathcal{N}:\mathcal{L}_A\to\mathcal{L}_B$, $D(\mathcal{N}(\rho)\|\mathcal{N}(\sigma))\le D(\rho\|\sigma)$. Intuitively, and in view of the operational interpretation of the relative entropy as a measure of distinguishability between $\rho$ and $\sigma$, it becomes clear that applying a quantum channel to the state never simplifies the discrimination task, thus leading to a reduction in the relative entropy. The above contraction is so fundamental to information theory that it is often taken as a requirement for any metric on quantum states to be 
called an information measure. Data processing can further be quantified through the use of so-called contraction coefficients, defined as
\begin{align}
    \eta(\mathcal{N}):=\sup_{\substack{\rho,\sigma\in\mathcal{S}_A \\\rho\ne \sigma}}\, \frac{D(\cN(\rho)\|\cN(\sigma))}{D(\rho\|\sigma)}. 
\end{align}
See also~\cite{hiai2016contraction, hirche2022contraction, hirche2023quantum} for additional properties and discussions.
The study of classical statistical problems under local privacy through the use of contraction coefficients was initiated in \cite{duchi2013local,duchi2018minimax}, where it was shown that $\SC_{\sym,\epsilon}(P,Q)=\Theta(\epsilon^{-2}\|P-Q\|_{\operatorname{TV}}^{-2})$, where $\|P-Q\|_{\operatorname{TV}}:=\frac{1}{2}\sum_{x\in\mathcal{X}}|P(x)-Q(x)|$ denotes the total variation between distributions $P$ and $Q$ over the alphabet $\mathcal{X}$. More recently, the following lower and upper bounds were derived in \cite[Lemma 2]{asoodeh2022contraction}:

\begin{align}\label{upperlowerSCe}
\max\left\{\left(\frac{e^\epsilon+1}{e^\epsilon-1}\right)^2\,\frac{\log(2.5)}{8 H_{\frac{1}{2}}(P,Q)},\frac{2}{25e^{-\epsilon}(e^\epsilon-1)^2\,\|P-Q\|_{\operatorname{TV}}^2}\right\} &\le \SC_{\sym,\epsilon}\left(P,Q,0.5,0.1\right) \\
&\le \left(\frac{e^\epsilon+1}{e^\epsilon-1}\right)^2\frac{2\log(5)}{\|P-Q\|^2_{\operatorname{TV}}}\,.
\end{align}
In the next sections, we aim at extending the above upper and lower bounds in to the framework of quantum states. All omitted proofs can be found in the appendix. 

\subsection{Symmetric hypothesis testing with uniform prior}
We begin with the symmetric setting, generalizing results from~\cite{asoodeh2022contraction} to the quantum setting. 
\subsubsection{Achieving LDP optimal sample complexity}

Our first main result is the following achievability bound for LDP hypothesis testing: 
\begin{theorem}[Achievability of LDP hypothesis testing] \label{Thm:Achievability}
For any two states $\rho,\sigma\in\mathcal{S}_A$, 
\begin{align}
\SC_{\sym,\epsilon}(\rho,\sigma) 
\leq
\left(\frac{e^\epsilon+1}{e^\epsilon-1}\right)^2\frac{2 \log 5}{ E_1(\rho\|\sigma)^2}\,.
\end{align}
\end{theorem}
\begin{proof}
We resort to Lemma~\ref{Lem:FvdG} and observe that,
\begin{align}
    E_1(\rho\|\sigma)^2 &\leq 1 - F(\rho,\sigma)^2 
    \leq 2 (1-F(\rho,\sigma)).
    \label{E1H12}
\end{align}
It remains to show that there exists an LDP algorithm $\cA:\mathcal{L}_A\to\mathcal{L}_B$ achieving the scaling for optimal sample complexity. Inspired by the use of classical binary algorithm in~\cite{asoodeh2022contraction}, we introduce the following channel for $\kappa\in[0,1]$:
\begin{align*}
    \cB(\cdot) = &|0\rangle\langle 0| \left( \kappa \tr(\{\rho\geq\sigma\} \,\cdot\,) + (1-\kappa) \tr(\{\rho<\sigma\} \,\cdot\, ) \right) \\
   + &|1\rangle\langle 1| \left( (1-\kappa) \tr(\{\rho\geq\sigma\} \,\cdot\,) + \kappa \tr(\{\rho<\sigma\} \,\cdot\, ) \right)\,.
\end{align*}

We can easily verify that 
\begin{align*}
    E_1(\cB(\rho)\|\cB(\sigma)) &= \frac12 \tr |\cB(\rho)-\cB(\sigma)| \\
    &= \frac12 | \kappa \tr(\{\rho\geq\sigma\} (\rho-\sigma)) + (1-\kappa) \tr(\{\rho<\sigma\} (\rho-\sigma) )| \\
    &\qquad + \frac12| (1-\kappa) \tr(\{\rho\geq\sigma\} (\rho-\sigma)) + \kappa \tr(\{\rho<\sigma\} (\rho-\sigma) ) | \\
    &= | (2\kappa-1) E_1(\rho\|\sigma) | .
\end{align*}
We now choose $\kappa:=\frac{e^\epsilon}{1+e^\epsilon}$. Therefore,
\begin{align*}
    E_1(\cB(\rho)\|\cB(\sigma)) 
    &= \left| \Big(\frac{2e^\epsilon}{1+e^\epsilon}-1\Big) E_1(\rho\|\sigma) \right| \\
    &=  \frac{e^\epsilon-1}{e^\epsilon+1} E_1(\rho\|\sigma)\,. 
\end{align*}
Combining with \eqref{E1H12}, we get 
\begin{align*}
    \sup_{\cA\in \LDP_\epsilon(A,B)} 1 - F (\cA(\rho)\|\cA(\sigma)) 
    &\geq \frac12\, \sup_{\cA\in \LDP_\epsilon(A,B)} E_1(\cA(\rho)\|\cA(\sigma))^2 \\
    &\geq \frac12\, \left(\frac{e^\epsilon-1}{e^\epsilon+1}\right)^2 E_1(\rho\|\sigma)^2.
\end{align*}
Combining with the upper bound in Proposition~\ref{prop:SC_original},
the result follows.
\end{proof}

\subsubsection{Optimality of LDP sample complexity}

Next, we aim at finding a lower bound for $\SC_\epsilon$. In fact, we derive two different ones.

\paragraph{Converse, Part I}

For our first lower bound, we make use of contraction coefficients for the trace distance and the relative entropy: given a quantum channel $\mathcal{N}:\mathcal{L}_A\to\mathcal{L}_B$, 
\begin{align*}
\eta_{\tr}(\mathcal{N})
&:=\sup_{\substack{\rho,\sigma\in\mathcal{S}_A\\\rho\ne \sigma}}\, \frac{\|\mathcal{N}(\rho-\sigma)\|_1}{\|\rho-\sigma\|_1} \\
&= \sup_{\Psi\perp\Phi} E_1(\cN(\Psi)\|\cN(\Phi))\,,
\end{align*}
where the second equality was shown in~\cite{ruskai1994beyond} with an optimization over orthogonal pure states. We start by proving a couple of Lemmas that generalize their classical analogues given in~\cite{asoodeh2022contraction}. The first one, whose proof we defer to Appendix \ref{appLem:H12-LDP-E1}, uses the tools from~\cite[Section 5.1]{hirche2023quantum}. 
\begin{lemma}\label{Lem:H12-LDP-E1}
    For any two states $\rho,\sigma\in\mathcal{S}_A$, 
    \begin{align}
        H_{\frac12}(\rho\|\sigma)\leq \left(\frac{(e^{\frac12D_{\infty}(\rho\|\sigma)}-1)^2}{e^{D_{\infty}(\rho\|\sigma)}-1} + \frac{(e^{\frac12D_{\infty}(\sigma\|\rho)}-1)^2}{e^{D_{\infty}(\sigma\|\rho)}-1}\right)E_1(\rho\|\sigma)
    \end{align}
\end{lemma}

This essentially follows the proof of~\cite[Proposition 5.2]{hirche2023quantum} but gives, by always choosing the tightest known bound, a slightly better result than~\cite[Corollary 5.5]{hirche2023quantum}. Next, we give a bound on the maximum output trace distance of LDP channels. See Appendix \ref{appLem:E1-LDP-UB} for a proof: 
\begin{lemma}\label{Lem:E1-LDP-UB}
    We have,
    \begin{align}
        \sup_{\cA\in\LDP_\epsilon}\sup_{\rho,\sigma\in\mathcal{S}_A} E_1(\cA(\rho)\|\cA(\sigma)) \leq \frac{e^{-\epsilon}(e^{\epsilon}-1)^2}{e^\epsilon-e^{-\epsilon}}. 
    \end{align}
\end{lemma}

We are now ready to prove the main result of this section (see Appendix \ref{appProp:eta-tr-LDP-UB} for details). 
\begin{proposition}\label{Prop:eta-tr-LDP-UB}
    For any quantum channel $\cA\in\LDP_\epsilon$, 
    \begin{align}
        \eta_{\tr}(\cA)  \leq \frac{(e^\epsilon-1)}{(e^\epsilon+1)}. 
    \end{align}
\end{proposition}

Alternatively we could have used the previously known bound~\cite{hirche2023quantumDP} $\eta_{\tr}(\cA) \leq 1-e^{-\epsilon}$, which would simplify the proof a lot, but unfortunately only gives the weaker bound $\frac{e^\epsilon-1}{e^\epsilon} \geq \frac{e^\epsilon-1}{e^\epsilon+1} \equiv \sqrt{\Upsilon_\epsilon}$. 

Now the proposition implies directly that 
\begin{align}
    \eta_{f}(\cA)  \leq \eta_{\tr}(\cA)  \leq \sqrt{\Upsilon_\epsilon}. \label{Eq:eta-f-upsilon-UB} 
\end{align}
which then implies 
\begin{align}
    \sup_{\cA\in \LDP_\epsilon} H_{\frac12}(\cA(\rho)\|\cA(\sigma)) \leq \sqrt{\Upsilon_\epsilon} H_{\frac12}(\rho\|\sigma). 
\end{align}
Note that, in the classical setting,~\cite{asoodeh2022contraction} proves the stronger 
\begin{align}
    \eta_{f}(\cA)  \leq \Upsilon_\epsilon, 
\end{align}
for operator convex $f$. 
There they start with (modulo notation)
\begin{align}
    \eta_{\operatorname{KL}}(\cA) \leq \sup_{x,x'} H_{\frac12}(\cA(\cdot|x)\|\cA(\cdot|x')) - \frac14 H_{\frac12}(\cA(\cdot|x)\|\cA(\cdot|x'))^2,
\end{align}
where $\eta_{\operatorname{KL}}(\cA)$ stands for the contraction coefficient for the Kullback–Leibler divergence. Proving a quantum version of this bound remains an interesting open problem.  Applying the above results to the problem of sample complexity, we get the following result.

\begin{theorem}[Converse I of LDP hypothesis testing] For any two states $\rho,\sigma\in\mathcal{S}_A$, 
\begin{align}
\sup_{\cA\in \LDP_\epsilon} H_{\frac12}(\cA(\rho)\|\cA(\sigma))\leq \sqrt{\Upsilon_\epsilon} H_{\frac12}(\rho\|\sigma)\,.
\end{align}
and hence, 
\begin{align}
    \SC_{\sym,\epsilon}(\rho,\sigma) \geq \frac{(e^\epsilon+1)\log 2.5}{2(e^\epsilon-1)H_{\frac12}(\rho\|\sigma)}. 
\end{align}
\end{theorem}
\begin{proof}
    This follows directly from Equation~\eqref{Eq:eta-f-upsilon-UB}. 
\end{proof}

\paragraph{Converse, Part II}

This proof is based on the $\chi^2$ divergence for which we have the following equivalent definitions~\cite{hirche2023quantum}, 
\begin{align*}
    &\chi^2(\rho\|\sigma)\equiv H_2(\rho\|\sigma) \\
    &= 2 \int_1^\infty ( E_\gamma(\rho\|\sigma) + \gamma^{-3} E_\gamma(\sigma\|\rho)) \d\gamma \\ 
    &= \int_0^\infty \tr[ (\rho-\sigma)(\sigma +s\Id)^{-1}(\rho-\sigma)(\sigma +s\Id)^{-1}] \,\d s .
\end{align*} 
The next Lemma is proved in Appendix \ref{appLem:chi2-UB-TV2} and is the core technical ingredient of this section. 
\begin{lemma}\label{Lem:chi2-UB-TV2}
For an arbitrary quantum channel $\cN$ and two input states $\rho,\sigma\in\mathcal{S}_A$, 
\begin{align*}
    \chi^2(\cN(\rho)\|\cN(\sigma)) \leq 2 E_1(\rho\|\sigma)^2 \max_{\Psi,\Phi} \chi^2(\cN(\Psi)\|\cN(\Phi)). 
\end{align*}
\end{lemma}
Note that this improves also the known classical result by a factor $2$. 
Now, similar to the classical case, we have 
\begin{align}
    \max_{\Psi,\Phi} \chi^2(\cN(\Psi)\|\cN(\Phi)) \leq \max_{\substack{\tau,\theta \\
    E_{e^\epsilon}(\tau\|\theta)=0 \\ E_{e^\epsilon}(\theta\|\tau)=0}} \chi^2(\tau\|\theta),
\end{align}
motivating the need for the following observation. 
\begin{lemma}
    We have
    \begin{align*}
        \max_{\substack{\tau,\theta \\
    E_{e^\epsilon}(\tau\|\theta)=0 \\ E_{e^\epsilon}(\theta\|\tau)=0}} E_1(\tau\|\theta) = e^{-\epsilon}\frac{(e^\epsilon-1)^2}{e^\epsilon-e^{-\epsilon}}, 
    \end{align*}
    and hence, 
    \begin{align*}
        \max_{\substack{\tau,\theta \\
    E_{e^\epsilon}(\tau\|\theta)=0 \\ E_{e^\epsilon}(\theta\|\tau)=0}} D_f(\tau\|\theta) \leq \frac{f(e^\epsilon)+e^\epsilon f(e^{-\epsilon})}{e^\epsilon-1}e^{-\epsilon}\frac{(e^\epsilon-1)^2}{e^\epsilon-e^{-\epsilon}}\,.
    \end{align*} 
\end{lemma}
\begin{proof}
    We start with the first statement for the trace distance. Define the measurement $\{\Pi_+,\Pi_-=\Id-\Pi_+\}$ and the probabilities $\{p=\tr\Pi_+\tau,1-p\}$ and $\{q=\tr\Pi_+\theta,1-q\}$. It is known that this measurement achieves the trace distance and hence $E_1(\tau\|\theta)=E_1(p\|q)$, where the right hand side is the classical binary total variation distance. Furthermore, we have $E_\gamma(p\|q)\leq E_\gamma(\tau\|\theta)$. Putting everything together we get,
    \begin{align}
        \max_{\substack{\tau,\theta \\
    E_{e^\epsilon}(\tau\|\theta)=0 \\ E_{e^\epsilon}(\theta\|\tau)=0}} E_1(\tau\|\theta) = \max_{\substack{0\leq p,q\leq 1 \\
    E_{e^\epsilon}(p\|q)=0 \\ E_{e^\epsilon}(q\|p)=0}} E_1(p\|q). 
    \end{align}
    It has therefore taken the same expression as in the classical case for which the solution was derived in~\cite[Equation (45)]{asoodeh2022contraction}. 
    The second statement for $f$-divergences follows from the reverse Pinsker inequality proven in~\cite[Proposition 5.2]{hirche2023quantum}. 
\end{proof}
As a special case for $f(x)=x^2-1$, we have 
\begin{align}
    \max_{\substack{\tau,\theta \\
    E_{e^\epsilon}(\tau\|\theta)=0 \\ E_{e^\epsilon}(\theta\|\tau)=0}} \chi^2(\tau\|\theta) \leq e^{-\epsilon}(e^\epsilon-1)^2, \label{Eq:chi2-LDP-UB}
\end{align}
which matches exactly the classical case. 
In summary, we have shown
\begin{align*}
   \max_{\cN\in\LDP_\epsilon} \chi^2(\cN(\rho)\|\cN(\sigma)) \leq 2e^{-\epsilon}(e^\epsilon-1)^2 E_1(\rho\|\sigma)^2 . 
    \end{align*}

\begin{theorem}[Converse II of LDP hypothesis testing]\label{Thm:ConverseII} For any two states $\rho,\sigma\in\mathcal{S}_A$, 
\begin{align*}
\SC_{\sym,\epsilon}\left(\rho,\sigma,0.5,\delta\right) &\geq \frac{(1-2\delta)^2}{e^{-\epsilon}(e^\epsilon-1)^2\,E_1(\rho\|\sigma)^2 } \\
\SC_{\sym,\epsilon}(\rho,\sigma,0.5,0.1) &\geq \frac{16}{25} \frac1{e^{-\epsilon}(e^\epsilon-1)^2\,E_1(\rho\|\sigma)^2 }
\end{align*}
\end{theorem}
\begin{proof}
    Consider the error probability of symmetric hypothesis testing, 
    \begin{align*}
        p_{e}(\rho,\sigma,0.5) = \frac12 ( 1 - E_1(\rho\|\sigma)).
    \end{align*}
    After rewriting that, we can continue with
    \begin{align*}
        2 ( 1 - 2p_{e}(\rho,\sigma,0.5))^2 = 2 E_1(\rho\|\sigma)^2 \leq D(\rho\|\sigma) \leq \chi^2(\rho\|\sigma), 
    \end{align*}
    where the first inequality is Pinsker's inequality and the second is from~\cite{temme2010chi}. Applying this to $n$ copies and using additivity of the relative entropy we get, 
    \begin{align*}
        &2 ( 1 - 2p_{e}(\cA(\rho)^{\otimes n},\cA(\sigma)^{\otimes n},0.5))^2  \\
        &\qquad \leq n \chi^2(\cA(\rho)\|\cA(\sigma)) \\
        &\qquad \leq 2 n E_1(\rho\|\sigma)^2 \max_{\Psi,\Phi} \chi^2(\cA(\Psi)\|\cA(\Phi)) \\
        &\qquad \leq 2 n E_1(\rho\|\sigma)^2 e^{-\epsilon}(e^\epsilon-1)^2, 
    \end{align*}
    where the second inequality is Lemma~\ref{Lem:chi2-UB-TV2} and the third Equation~\eqref{Eq:chi2-LDP-UB}. Choosing the error probability as $0.1$, this is gives
    \begin{align*}
        n \geq \frac{16}{25} \frac1{e^{-\epsilon}(e^\epsilon-1)^2\,E_1(\rho\|\sigma)^2 },
    \end{align*}
    from which the claim follows. 
\end{proof} 
Note that this is a factor 8 better than the classical result in~\cite{asoodeh2022contraction}: A factor 2 because of the improvement in Lemma~\ref{Lem:chi2-UB-TV2} and a factor 4 because of a suboptimal use of Pinskers inequality in~\cite[Lemma 2]{asoodeh2022contraction}.

\subsection{Symmetric hypothesis testing with arbitrary prior}

Lower bound with differential privacy. 
\begin{theorem}\label{lowpboundSCfirst}
    We have, for any $p\in(0,1)$ and $\delta = \frac{p}{4}$,
    \begin{align}
        \SC_{\sym,\epsilon}(\rho,\sigma,p, \delta) &\geq \frac3{16}\frac{e^\epsilon+1}{e^\epsilon-1} \frac{p\log{p^{-1}}}{\JS_p(\rho\|\sigma)}\geq \frac{3\sqrt{2}}{16} \frac{e^\epsilon+1}{e^\epsilon-1} \frac{\log2}{\lambda  H_{1-\lambda}(\rho\|\sigma)},  \\
        \SC_{\sym,\epsilon}(\rho,\sigma,p,\delta) &\geq \frac9{16}\frac{1}{e^{-\epsilon}(e^\epsilon-1)E_1^2(\rho\|\sigma)}.  
    \end{align}
\end{theorem}
\begin{proof}
    The proofs follow similarly to the $p=\frac12$ case discussed before. The first is simply by the upper bound on the trace distance contraction coefficient given in Proposition~\ref{Prop:eta-tr-LDP-UB}. 
    The second, follows from 
    \begin{align}
        p_e^*(\rho^{\otimes n},\sigma^{\otimes n},p) & = p - p E_\frac{1-p}{p}(\rho^{\otimes n}\|\sigma^{\otimes n}),
    \end{align}
    which implies, 
    \begin{align}
        2\left(1-\frac1p p_e^*(\cA(\rho)^{\otimes n},\cA(\sigma)^{\otimes n},p)\right)^2 
        &= 2 E^2_\frac{1-p}{p}(\cA(\rho)^{\otimes n}\|\cA(\sigma)^{\otimes n}) \\
        &\leq 2 E^2_1(\cA(\rho)^{\otimes n}\|\cA(\sigma)^{\otimes n}). 
    \end{align}
    Following the proof of Theorem~\ref{Thm:ConverseII}, this leads to 
    \begin{align} \label{eq:lowpboundSCfirst}
        n \geq \frac{(1-\frac1p p_e^*(\cA(\rho)^{\otimes n},\cA(\sigma)^{\otimes n},p))^2}{e^{-\epsilon}(e^\epsilon-1)E^2_1(\rho\|\sigma)} \geq \frac{(1-\frac1p \delta)^2}{e^{-\epsilon}(e^\epsilon-1)E^2_1(\rho\|\sigma)}.
    \end{align}
    The result then follows by choosing $\delta \leq\frac{p}4$. 
\end{proof}

As an alternative to Proposition~\ref{Prop:prior-upper}, the following sample complexity upper bound was shown in~\cite[Theorem 7]{Wilde2024_review}, 
\begin{align}
    \SC_{\sym}(\rho,\sigma,p,\delta) \leq \inf_{s\in[0,1]} \frac{\log{\frac{p^s(1-p)^{1-s}}{\delta}}}{-\log \widebar Q_s(\rho\|\sigma)}. 
\end{align}
We can use this to give the following LDP bound.
\begin{theorem}\label{upboundSCeps}
    We have,
    \begin{align} \label{eq:upboundSCeps}
        \SC_{\sym,\epsilon}(\rho,\sigma,p,\delta) &\leq 
         \left(\frac{e^\epsilon+1}{e^\epsilon-1}\right)^2 \frac{2\log{\delta^{-1}}}{ E_1(\rho\|\sigma)^2}, \\
        \SC_{\sym,\epsilon}\left(\rho,\sigma,p,\frac{p}4\right) &\leq 
         \left(\frac{e^\epsilon+1}{e^\epsilon-1}\right)^2 \frac{2\log{p^{-1}}+2\log{4}}{ E_1(\rho\|\sigma)^2}. 
    \end{align}
\end{theorem}
\begin{proof}
    Let $s^\star=\argmin_{s\in[0,1]} \widebar Q_{1-s}(\rho\|\sigma)$ and $Q_{\min}(\rho\|\sigma)=\min_{s\in[0,1]} \widebar Q_s(\rho\|\sigma)$. then, 
    \begin{align}
        \SC_{\sym}(\rho,\sigma,p,\delta) \leq \inf_{s\in[0,1]} \frac{\log{\frac{p^{1-s}(1-p)^{s}}{\delta}}}{-\log \widebar Q_{1-s}(\rho\|\sigma)} \leq
         \frac{\log{\frac{p^{1-s^\star}(1-p)^{s^\star}}{\delta}}}{-\log \widebar Q_{\min}(\rho\|\sigma)}         \leq
         \frac{\log{\frac{p^{1-s^\star}(1-p)^{s^\star}}{\delta}}}{1- \widebar Q_{\min}(\rho\|\sigma)} 
    \end{align}
    Now, from~\cite{audenaert2012comparisons}, we have
    \begin{align}
        \widebar Q_{\min}(\rho\|\sigma) \leq \widebar Q_{\frac12}(\rho\|\sigma) \leq F(\rho,\sigma). 
    \end{align}
    Together with, 
    \begin{align}
        1-F(\rho,\sigma) \geq \frac12 (1-F(\rho,\sigma)^2)\geq \frac12 E_1(\rho\|\sigma)^2, 
    \end{align}
    this results in
    \begin{align}
        \SC_{\sym}(\rho,\sigma,p,\delta) \leq 
         \frac{2\log{\frac{p^{1-s^\star}(1-p)^{s^\star}}{\delta}}}{ E_1(\rho\|\sigma)^2}
    \end{align}
    Following the proof of Theorem~\ref{Thm:Achievability}, we have
    \begin{align}
    \sup_{\cA\in \LDP_{\epsilon}} E_1(\cA(\rho)\|\cA(\sigma)) 
    \geq  \frac{e^\epsilon-1}{e^\epsilon+1} E_1(\rho\|\sigma)\,, 
\end{align}
leading to
    \begin{align}
        \SC_\epsilon(\rho,\sigma,p,\delta) \leq 
         \left(\frac{e^\epsilon+1}{e^\epsilon-1}\right)^2 \frac{2\log{\frac{p^{1-s^\star}(1-p)^{s^\star}}{\delta}}}{ E_1(\rho\|\sigma)^2}. 
    \end{align}
The first result then follows from $p^{1-s^\star}(1-p)^{s^\star}\leq 1$. The second by specializing to $\delta=\frac{p}4$.
\end{proof}

Using instead the bound in Proposition~\ref{Prop:prior-upper}, alternative bounds on the sample complexity can be obtained by relating the fidelity to the Hellinger divergence.

\begin{lemma}
For any $\alpha\in[0,1]$ and any two states $\rho,\sigma$, setting $\beta=\min\{\alpha,1-\alpha\}$,
\begin{align*}
\widebar{H}_\alpha(\rho\|\sigma)\ge \frac{1-F(\rho,\sigma)^{2\beta}}{\beta}\ge  \frac{1-\big(1- E_1(\rho\|\sigma)^2\big)^{\beta}}{\beta}\,.
\end{align*}
\end{lemma}
\begin{proof}
By symmetry, we can assume that $\alpha\in[1/2,1]$ without loss of generality. The proof follows by an interpolation argument: we consider the function $f:S\to \mathbb{C}$ on the strip $S=\{z\in\mathbb{C}|\operatorname{Re}(z)\in [0,\frac{1}{2}]\}$ defined as
\begin{align*}
f(z):=\tr(\rho^{z}\sigma^{1-z}).
\end{align*}
We have that $\left|f\left(1+it\right)\right|\le 1$ for all $t\in\mathbb{R}$ and 
\begin{align*}
\left|f\left(\frac{1}{2}+it\right)\right|=\Big|\tr\Big(\rho^{\frac{1}{2}+it}\sigma^{\frac{1}{2}-it}\Big)\Big|\le F(\rho,\sigma)\,.
\end{align*}
Therefore, by Hadamard three-lines theorem, we directly get that
\begin{align*}
f(\alpha)\le F(\rho,\sigma)^{2(1-\alpha)} \,.
\end{align*}
Therefore, 
\begin{align*}
\widebar{H}_\alpha(\rho\|\sigma)=\frac{1}{1-\alpha}\left(1-\tr\rho^\alpha\sigma^{1-\alpha}\right)\ge \frac{1-F(\rho,\sigma)^{2(1-\alpha)}}{1-\alpha}\ge  \frac{1-\big(1-E_1(\rho\|\sigma)^2\big)^{1-\alpha}}{1-\alpha}\,,
\end{align*}
where the last inequality follows from Fuchs-van-de-Graaf's inequality, see Lemma \ref{Lem:FvdG}.
\end{proof}

As a result, we can derive the following bound on the private sample complexity. 

\begin{corollary}
For $\lambda=\frac{\log2}{2\log{p^{-1}}}$ and $\beta=\min\{\lambda,1-\lambda\}$, we have
\begin{align*}
\SC_{\sym,\epsilon}(\rho,\sigma,p,\delta)\le \frac{\beta\log{\left(p^{1-\lambda}(1-p)^\lambda/\delta\right)}}{\lambda\left(1-\left(1-\left(\frac{e^\epsilon-1}{e^\epsilon+1}\right)^2E_1(\rho\|\sigma)^2\right)^{\beta}\right)}\,.
\end{align*}
\end{corollary}

\begin{proof}
By Proposition \ref{Prop:prior-upper} we have, for $\lambda=\frac{\log2}{2\log{p^{-1}}}$ and $\beta=\min\{\lambda,1-\lambda\}$,
    \begin{align}
        \SC_{\sym}(\rho,\sigma,p,\delta) &\leq \frac{\log{\left(p^{1-\lambda}(1-p)^\lambda/\delta\right)}}{\lambda \widebar H_{1-\lambda}(\rho\|\sigma)}\le \frac{\beta\log{\left(p^{1-\lambda}(1-p)^\lambda/\delta\right)}}{\lambda\big(1-\big(1-E_1(\rho\|\sigma)^2\big)^{\beta}\big)} \,.
    \end{align}
     Following once again the proof of Theorem~\ref{Thm:Achievability}, we have
    \begin{align}
    \sup_{\cA\in \LDP} E_1(\cA(\rho)\|\cA(\sigma)) 
    \geq  \frac{e^\epsilon-1}{e^\epsilon+1} E_1(\rho\|\sigma)\,
\end{align} 
    the result follows.
\end{proof}

\subsection{Asymmetric hypothesis testing} \label{sec:asym}

In the previous sections, we studied the sample complexity for the \emph{symmetric hypothesis testing} (which is also called the \emph{Bayesian hypothesis testing}), where the prior probabilities for the hypotheses are provided and fixed.
In the following, we discuss the relation to the \emph{asymmetric hypothesis testing} (which is also called the \emph{prior-free hypothesis testing}), where the aim is to balance the so-called type-I and type-II errors without knowing priors.

For the asymmetric setting, 
we denote $\SC_{\asym}(\rho,\sigma,\alpha,\beta )$ as the smallest integer $n$ such that there exists a test $0\leq T_n \leq I$ satisfying
\begin{align} \label{eq:defn:asym}
    \tr\left( \rho^{\otimes n} (I - T_n)  \right) \leq \alpha, \quad
    \tr \left( \sigma^{\otimes n} T_n \right) \leq \beta.
\end{align}
Similar to the definition for the symmetric setting, \eqref{eq:defn:sym},
we use the subscript `PF' to designate the sample complexity to the asymmetric setting.

Ref.~\cite[Corollary 4.8]{pensia2024sample} show  that the sample complexities for both settings are equivalent up to a multiplicative constant in the classical scenario.
Here, we argue that the statement naturally extends to the quantum scenario:

\begin{theorem} \label{theorem:relation}
    Let $\alpha \in (0,1/2]$ and $\delta \leq \alpha/4$.
    Then $\SC_{\sym}(\rho,\sigma,\alpha,\delta) \asymp \SC_{\asym}(\rho,\sigma,\frac{\delta}{\alpha}, \frac{\delta}{1-\alpha} )$, where ``$\asymp$'' means equality up to a universal constant.

    Similarly, let $\alpha, \beta \in (0,1/8]$ satisfying $\beta \leq \alpha$.
    Then $\SC_{\asym}(\rho,\sigma,\alpha, \beta)
    \asymp \SC_{\sym}(\rho,\sigma,\frac{\beta}{\alpha+\beta},\frac{\alpha\beta}{\alpha+\beta}) $.
\end{theorem}

Below we explain why the ideas in \cite{pensia2024sample} holds for the quantum scenario (Theorem~\ref{theorem:relation}).
First, \cite[Claim 4.6]{pensia2024sample} shows that both the complexities can be related.
By definitions given in \eqref{eq:defn:sym} and \eqref{eq:defn:asym}, such an assertion of the relation holds regardless of the classical or quantum scenario.
Hence, the following hold:
\begin{align} \label{eq:relation1}
\SC_{\sym}\left(\rho,\sigma,\frac{\beta}{\alpha+\beta},\frac{2\alpha\beta}{\alpha+\beta}\right)
\leq 
\SC_{\asym}\left(\rho,\sigma,\alpha,\beta\right)
\leq 
\SC_{\sym}\left(\rho,\sigma,\frac{\beta}{\alpha+\beta},\frac{\alpha\beta}{\alpha+\beta}\right), \quad \alpha,\beta \in (0,1),
\end{align}
and
\begin{align} \label{eq:relation2}
\SC_{\asym}\left(\rho,\sigma,\frac{\delta}{\alpha},\frac{\delta}{1-\alpha}\right)
\leq 
\SC_{\sym}\left(\rho,\sigma,\alpha,\beta\right)
\leq 
\SC_{\sym}\left(\rho,\sigma,\frac{\delta}{2\alpha},\frac{\delta}{2(1-\alpha)}\right), \quad \alpha,\delta \in (0,1),
\end{align}

Second, the upper and lower bounds in the above inequalities can reversed by a multiplicative factor
as claimed by \cite[Proposition 4.7]{pensia2024sample}.
This can be done by resorting to a useful classical argument (see also \cite[Fact 4.9]{pensia2024sample}):
\begin{lemma}[Repetition to boost success probability] \label{lemma:repetition}
    Let $\rho$ be a quantum state.
    Consider a test $0\leq T_n \leq I$ satisfying
    $\tr\left( \rho^{\otimes n} (I - T_n) \right) \leq \alpha $ for $\alpha \leq 1/4$.
    Then, there exists a modified test $0\leq T_{nr} \leq I$ defined by the majority of $r$-rounds of independent tests via $T_n$ such that the boosted procedure is at most $\alpha' \leq \alpha$, i.e.,
    \begin{align}
    \tr\left( \rho^{\otimes nr} (I - T_{nr}) \right) \leq \alpha'
    \end{align}
    provided that $r \geq \frac{2^5 \log (1/\alpha')}{\log (1/\alpha) }$.
\end{lemma}
Since majority vote is a classical post-processing, thereby classical concentration such as Bennet's inequality \cite[Fact A.1]{pensia2024sample}, \cite[Theorem 2.9.2]{Vershynin2018book}.
Using Lemma~\ref{lemma:repetition}, \cite[Proposition 4.7]{pensia2024sample} straightforwardly extends to the quantum scenario:
For $\alpha_1, \alpha_2 \in (0,1/2]$, $\delta_1 \in (0,\alpha_1/4)$, and $\delta_2 \in (0, \alpha_2/4)$,
\begin{align} \label{eq:reverse_control1}
\frac{ \SC_{\sym}(\rho,\sigma, \alpha_1, \delta_1)  }{ \SC_{\sym}(\rho,\sigma, \alpha_2, \delta_2) } \leq O\left( \max\left\{ 1, \frac{\log(\alpha_1/\delta_1)}{ \log(\alpha_2/\delta_2)  }, \frac{ \log (1/\delta_1)}{ \log (1/\delta_2) } \right\} \right) .
\end{align}
Similarly, for $\alpha_1, \beta_1, \alpha_2, \beta_2 \in (0,1/4]$, 
\begin{align} \label{eq:reverse_control2}
\frac{ \SC_{\asym}(\rho,\sigma, \alpha_1, \beta_1)  }{ \SC_{\asym}(\rho,\sigma, \alpha_2, \beta_2) } \leq O\left( \max\left\{ 1, \frac{\log(1/\alpha_1)}{ \log(1/\alpha_2)  }, \frac{ \log (1/\beta_1)}{ \log (1/\beta_2) } \right\} \right) .
\end{align}
Then, inequalities \eqref{eq:relation1}, \eqref{eq:relation2}, \eqref{eq:reverse_control1}, and \eqref{eq:reverse_control2} together imply Theorem~\ref{theorem:relation} via the argument in \cite[Corollary 4.8]{pensia2024sample}.

\bigskip
Next, we are at a position to move to the private setting for asymmetric hypothesis testing.
By defining
\begin{align}
  \SC_{\epsilon, \asym}(\rho,\sigma, \alpha, \beta) 
  &:= \inf_{\cA\in \LDP_\epsilon}\,\SC_{\asym} (\mathcal{A}(\rho),\mathcal{A}(\sigma),\alpha, \beta),
\end{align}
then the equivalence relation, Theorem~\ref{theorem:relation}, and the bounds  established for the symmetric setting in previous sections
give us the following bounds for the $\epsilon$-locally differentially private asymmetric hypothesis testing.
\begin{corollary}
Let $\alpha,\beta \in (0,1/8]$. Then,
    \begin{align}
        \frac{(1-\alpha)^2}{e^{-\epsilon}(e^\epsilon-1)E_1(\rho\|\sigma)^2}
        \lesssim \SC_{\asym,\epsilon }\left(\rho,\sigma,\alpha,\beta\right) &\lesssim
         \left(\frac{e^\epsilon+1}{e^\epsilon-1}\right)^2 \frac{\log{ \frac{\alpha+\beta}{\alpha\beta}  }}{ E_1(\rho\|\sigma)^2}\,,
\end{align}
where `$\lesssim$' means an inequality up to a universal constant.
\end{corollary}
\begin{proof}
Theorem~\ref{theorem:relation} and the bounds \eqref{eq:upboundSCeps} and \eqref{eq:lowpboundSCfirst} with $p = \frac{\beta}{\alpha+\beta}$ and $\delta = \frac{\alpha\beta}{\alpha+\beta} \leq \frac14$
give us the desired bounds.
\end{proof}

\section*{Acknowledgement}
HC is supported by Grants No.~NSTC 112-2636-E-002-009, No.~NSTC 112-2119-M-007-006, No.~NSTC 112-2119-M-001-006, No.~NSTC 112-2124-M-002-003, No.~NTU-112V1904-4, NTU-CC-113L891605, and NTU-113L900702.

\bibliographystyle{abbrv}
\bibliography{bib}

\newpage
\onecolumn
\appendix
\section{Proofs}
In this section we provide the proofs missing in the main text.

\subsection{Proof of Lemma~\ref{Lem:H12-LDP-E1}}\label{appLem:H12-LDP-E1}
\begin{proof}
    By specializing Equation~\eqref{eq:H_alpha} we get, 
        \begin{align}
            &H_{\frac12}(\rho\|\sigma) \\
            &= \frac12\int_1^\infty\gamma^{-\frac32}\Big(E_\gamma(\rho\|\sigma)+E_\gamma(\sigma\|\rho)\,\Big)d\gamma \\
            &= \frac12\int_1^{e^{D_\infty(\rho\|\sigma)}}\gamma^{-\frac32}E_\gamma(\rho\|\sigma)d\gamma + \frac12\int_1^{e^{D_\infty(\sigma\|\rho)}}\gamma^{-\frac32}E_\gamma(\sigma\|\rho)d\gamma \\
            &\leq \frac12\int_1^{e^{D_\infty(\rho\|\sigma)}}\gamma^{-\frac32}\frac{e^{D_\infty(\rho\|\sigma)}-\gamma}{e^{D_\infty(\rho\|\sigma)}-1}E_1(\rho\|\sigma)d\gamma + \frac12\int_1^{e^{D_\infty(\sigma\|\rho)}}\gamma^{-\frac32}\frac{e^{D_\infty(\sigma\|\rho)}-\gamma}{e^{D_\infty(\sigma\|\rho)}-1}E_1(\sigma\|\rho)d\gamma \\
            &= \frac{(e^{\frac12D_\infty(\rho\|\sigma)}-1)^2}{e^{D_\infty(\rho\|\sigma)}-1}E_1(\rho\|\sigma) + \frac{(e^{\frac12D_\infty(\sigma\|\rho)}-1)^2}{e^{D_\infty(\sigma\|\rho)}-1}E_1(\sigma\|\rho),
        \end{align}
        where the second equality holds because $E_\gamma(\rho\|\sigma)=0$ for $\gamma\geq D_\infty(\rho\|\sigma)$, the inequality by convexity of $E_\gamma$ and the last equality by evaluating the integral. 
\end{proof}

\subsection{Proof of Lemma~\ref{Lem:E1-LDP-UB}}\label{appLem:E1-LDP-UB}
\begin{proof}
    We follow closely the classical proof in~\cite[Appendix C]{asoodeh2022contraction}. First, observe,
    \begin{align}
        \sup_{\cA}\sup_{\rho,\sigma} E_1(\cA(\rho)\|\cA(\sigma)) 
        \leq \sup_{\substack{\rho,\sigma \\ E_{e^\epsilon}(\rho\|\sigma)=0 \\ E_{e^\epsilon}(\sigma\|\rho)=0}} E_1(\rho\|\sigma), \label{Eq:trace-LDP-1}
    \end{align}
    which holds because the channel outputs are always close in the corresponding Hockey-stick divergence by definition. Next, define the measurement, 
    \begin{align}
        \cM(\cdot) = |0\rangle\langle 0| \tr\{P_+ \,\cdot\,\} + |1\rangle\langle 1|\tr\{(\id-P_+)\,\cdot\,\}, 
    \end{align}
    where 
    $P_+=\{\rho\geq 0\sigma\}$
    It is now easy to check that 
    \begin{align}
        E_1(\rho\|\sigma)=E_1(\cM(\rho)\|\cM(\sigma)), 
    \end{align}
    and, by data processing, 
    \begin{align}
         E_\gamma(\rho\|\sigma)=0 \quad\Rightarrow\quad E_\gamma(\cM(\rho)\|\cM(\sigma))=0. 
    \end{align}
    The output of the measurement can simply be seen as a binary probability distribution and hence, 
    \begin{align}
        \sup_{\substack{\rho,\sigma \\ E_{e^\epsilon}(\rho\|\sigma)=0 \\ E_{e^\epsilon}(\sigma\|\rho)=0}} E_1(\rho\|\sigma) 
        = \sup_{\substack{p,q \\ E_{e^\epsilon}(p\|q)=0 \\ E_{e^\epsilon}(q\|p)=0}} E_1(p\|q) =  \frac{e^{-\epsilon}(e^{\epsilon}-1)^2}{e^\epsilon-e^{-\epsilon}}, 
    \end{align}
    where $p,q$ are said binary probability distributions and the final equality was shown in~\cite[Equation (43)]{asoodeh2022contraction}. 
    Inserting this back into Equation~\eqref{Eq:trace-LDP-1} gives the claimed result. 
\end{proof}

\subsection{Proof of Proposition~\ref{Prop:eta-tr-LDP-UB}}\label{appProp:eta-tr-LDP-UB}

\begin{proof}
    We start with, 
    \begin{align}
    \eta_{\tr}(\cA) 
    &= \sup_{\Psi\perp\Phi} E_1(\cA(\Psi)\|\cA(\Phi)) \\
    &\leq \sup_{\Psi\perp\Phi} \sqrt{1 - F(\cA(\Psi)\|\cA(\Phi))^2}, 
    \\
    &\leq \sup_{\Psi\perp\Phi} \sqrt{H_{\frac12}(\cA(\Psi)\|\cA(\Phi)) - \frac14 H_{\frac12}(\cA(\Psi)\|\cA(\Phi))^2}, \label{Eq:eta-tr-H12-UB}
    \end{align}
    which follows from the Fuchs-van-de-Graaf inequality \cite{FG99} and then the first inequality, $1-F(\cdot,\cdot) \leq \frac{ H_{\frac12}(\cdot,\cdot)}{2}$, in Lemma~\ref{Lem:FvdG}.
    As observed in~\cite{asoodeh2022contraction}, due to the monotonicity of $t\rightarrow t(1-\frac14 t)$ in $[0,2]$, it is sufficient to continue with
    \begin{align}
        \sup_{\Psi\perp\Phi} H_{\frac12}(\cA(\Psi)\|\cA(\Phi))  \leq 2\frac{(e^{\frac12\epsilon}-1)^2}{e^{\epsilon}-1} \sup_{\Psi\perp\Phi} E_1(\cA(\Psi)\|\cA(\Phi)),
    \end{align}
    which follows from Lemma~\ref{Lem:H12-LDP-E1} because $\cA$ being LDP implies that $D_\infty(\cA(\rho)\|\cA(\sigma))\leq \epsilon$ and $D_\infty(\cA(\sigma)\|\cA(\rho))\leq \epsilon$. By then applying Lemma~\ref{Lem:E1-LDP-UB}, we can further bound this by, 
    \begin{align}
        \sup_{\Psi\perp\Phi} H_{\frac12}(\cA(\Psi)\|\cA(\Phi))  
        &\leq 2\frac{(e^{\frac12\epsilon}-1)^2}{e^{\epsilon}-1} \frac{e^{-\epsilon}(e^{\epsilon}-1)^2}{e^\epsilon-e^{-\epsilon}} \\
        &= 2\frac{(e^{\frac12\epsilon}-1)^2 (1-e^{-\epsilon})}{e^\epsilon-e^{-\epsilon}}. 
    \end{align}
    Inserting this back into Equation~\eqref{Eq:eta-tr-H12-UB} gives, after some calculation, 
    \begin{align}
        \eta_{\tr}(\cA) \leq \frac{e^\epsilon-1}{e^\epsilon+1}, 
    \end{align}
    which is was we set out to prove. 
\end{proof}

\subsection{Proof of Lemma~\ref{Lem:chi2-UB-TV2}}\label{appLem:chi2-UB-TV2}
\begin{proof}
Let $\sigma=\sum_i p_i |i\rangle\langle i|$. We observe the following, 
\begin{align}
    &\chi^2(\cN(\rho)\|\cN(\sigma))\\
    &= \int_0^\infty \tr[ (\cN(\rho)-\cN(\sigma))(\cN(\sigma) +s\Id)^{-1}(\cN(\rho)-\cN(\sigma))(\cN(\sigma) +s\Id)^{-1}] \d s \\
    &\leq \int_0^\infty \tr[\sum_{i,j}p_ip_j (\cN(\rho)-\cN(\sigma))(\cN(|i\rangle\langle i|) +s\Id)^{-1}(\cN(\rho)-\cN(\sigma))(\cN(|j\rangle\langle j|) +s\Id)^{-1}] \d s \\
    &\leq \sum_{i} p_i\int_0^\infty \tr[ (\cN(\rho)-\cN(\sigma))(\cN(|i\rangle\langle i|) +s\Id)^{-1}(\cN(\rho)-\cN(\sigma))(\cN(|i\rangle\langle i|) +s\Id)^{-1}] \d s \\
    &\leq \max_i \int_0^\infty \tr[ (\cN(\rho)-\cN(\sigma))(\cN(|i\rangle\langle i|) +s\Id)^{-1}(\cN(\rho)-\cN(\sigma))(\cN(|i\rangle\langle i|) +s\Id)^{-1}] \d s,
\end{align}
where the first inequality is operator convexity of $x^{-1}$ and the second for $x^2$ 
Now, define the replacer channel
\begin{align}
    \cR(\cdot) = \cN(|i\rangle\langle i|) \tr(\cdot), 
\end{align}
and set $\rho-\sigma=X=X_+-X_-$, with 
\begin{align}
    X_+=\sum_n p_n |n\rangle\langle n|, \qquad X_-=\sum_m p_m |m\rangle\langle m|. 
\end{align}
Finally define 
\begin{align}
    \hat X_+ = \frac{X_+}{\tr X_+}, \qquad \hat X_- = \frac{X_-}{\tr X_-},
\end{align}
where $\tr X_+ = \tr X_- = E_1(\rho\|\sigma)$. 
Note that
\begin{align}
    \cN(\rho-\sigma) &= (\cN-\cR)(\rho-\sigma)=(\cN-\cR)(X_+-X_-)=(\cN-\cR)((\tr X_+)\hat X_+-(\tr X_-)\hat X_-) \\
    &= E_1(\rho\|\sigma) (\cN-\cR)(\hat X_+-\hat X_-). 
\end{align}
With this, we get
\begin{align}
    \chi^2(\cN(\rho)\|\cN(\sigma)) &\leq E_1(\rho\|\sigma)^2 \max_i \int_0^\infty \tr\left[ \left( ((\cN-\cR)(\hat X_+-\hat X_-))(\cN(|i\rangle\langle i|) +s\Id)^{-1} \right)^2\right] \d s \\
    &\leq E_1(\rho\|\sigma)^2 \max_i \left(\int_0^\infty \tr\left[ \left( ((\cN-\cR)(\hat X_+))(\cN(|i\rangle\langle i|) +s\Id)^{-1} \right)^2\right] \d s \right.\\
    &\quad \left.+ \int_0^\infty \tr\left[ \left( ((\cN-\cR)(\hat X_-))(\cN(|i\rangle\langle i|) +s\Id)^{-1} \right)^2\right] \d s\right),  \\
    &\leq E_1(\rho\|\sigma)^2 \max_i \left(\sum_n p_n \int_0^\infty \tr\left[ \left( ((\cN-\cR)(|n\rangle\langle n|))(\cN(|i\rangle\langle i|) +s\Id)^{-1} \right)^2\right] \d s \right.\\
    &\quad \left.+ \sum_m p_m \int_0^\infty \tr\left[ \left( ((\cN-\cR)(|m\rangle\langle m|))(\cN(|i\rangle\langle i|) +s\Id)^{-1} \right)^2\right] \d s\right), \\
    &\leq 2 E_1(\rho\|\sigma)^2 \max_{\Psi,\Phi}  \int_0^\infty \tr\left[ \left( (\cN(\Psi)-\cN(\Phi))(\cN(\Phi) +s\Id)^{-1} \right)^2\right] \d s \\
    &= 2 E_1(\rho\|\sigma)^2 \max_{\Psi,\Phi} \chi^2(\cN(\Psi)\|\cN(\Phi)),
\end{align}
where the second inequality is by dropping all negative terms, the third by convexity and the forth by optimizing over general pure states. 
\end{proof}
\end{document}